\documentclass{article}
\pdfoutput=1
\usepackage{amsmath}
\usepackage{amsfonts}
\usepackage{amsthm}
\usepackage{kantlipsum}
\usepackage{amssymb}
\usepackage{graphicx}
\usepackage[colorlinks,citecolor=blue,urlcolor=blue]{hyperref}
\usepackage{multirow}
\usepackage{multicol}
\usepackage{lipsum}
\usepackage{array}
\usepackage{float}
\usepackage{booktabs}
\usepackage{indentfirst}
\usepackage{chngcntr}
\usepackage[mathscr]{euscript}
\usepackage{mathrsfs}
\usepackage{color}
\usepackage{hyperref}
\usepackage{marginnote}
\usepackage{pdfpages}

\newtheorem{thm}{Theorem}[section]

\newtheorem{lem}[thm]{Lemma}
\newtheorem{prop}[thm]{Proposition}
\theoremstyle{definition}

\theoremstyle{remark}

\newtheorem*{acc*}{}
\theoremstyle{plain}

\usepackage[nonatbib,final]{nips_2016}

\usepackage[utf8]{inputenc} 
\usepackage[T1]{fontenc}    
\usepackage{hyperref}       
\usepackage{url}            
\usepackage{booktabs}       
\usepackage{amsfonts}       
\usepackage{nicefrac}       
\usepackage{microtype}      

\title{Estimation of Vertex Degrees in a Sampled Network}

\author{
  Apratim Ganguly\thanks{Currently at Natera Inc.} \\
  Department of Mathematics and Statistics\\
  Boston University\\
  Boston, MA 02212 \\
  \texttt{apratimganguly@gmail.com} \\
  \And
  Eric Kolaczyk \\
  Department of Mathematics and Statistics\\
  Boston University \\
  Boston, MA 02212\\
  \texttt{kolaczyk@bu.edu} \\
}
\begin{document}

\maketitle
\begin{abstract}
 The need to produce accurate estimates of vertex degree in a large network, based on observation of a subnetwork, arises in a number of practical settings. We study a formalized version of this problem, wherein the goal is, given a randomly sampled subnetwork from a large parent network, to estimate the actual degree of the sampled nodes. Depending on the sampling scheme, trivial method of moments estimators (MMEs) can be used. However, the MME is not expected, in general, to use all relevant network information. In this study, we propose a handful of novel estimators derived from a risk-theoretic perspective, which make more sophisticated use of the information in the sampled network. Theoretical assessment of the new estimators characterizes under what conditions they can offer improvement over the MME, while numerical comparisons show that when such improvement obtains, it can be substantial.  Illustration is provided on a human trafficking network.
\end{abstract}

\section{Introduction}\label{sec:introduction}
Frequently it is the case in the study of real-world complex networks that we observe essentially a sample from a larger network.  There are many reasons why sampling in networks is often unavoidable -- and, in some cases, even desirable.   Sampling, for example, has long been a necessary part of studying Internet topology~\cite{crovella2006internet}.  Similarly, its role has been long-recognized in the context of biological networks, e.g., protein-protein interaction \cite{Jeong2000,Maslov2002,Qin2003}, gene regulation \cite{Noort2004} and metabolic networks \cite{Jeong2000}.  Finally, in recent years, there has been intense interest in the use of sampling for monitoring online social media networks.  See~\cite{Zhang2015}, for example, for a representative list of articles in this latter domain.  Given a sample from a network, a fundamental statistical question is how the sampled network statistics be used to make inferences about the parameters of the underlying global network. Parameters of interest in the literature include (but are by no means limited to) degree distribution, density, diameter, clustering coefficient, and number of connected components. For seminal work in this direction, see~\cite{Frank1980,Frank1981}.

In this paper, we propose potential solutions to an estimation problem that appears to have received significantly less attention in the literature to date -- the estimation of the degrees of individual sampled nodes. Degree is one of the most fundamental of network metrics, and is a basic notion of node-centrality.  Deriving a good estimate of the node degree, in turn, can be helpful in estimating other global parameters, as many such parameters can be viewed as functions that include degree as an argument.  While a number of methods are available to estimate the full degree distribution under network sampling (e.g., \cite{Stumpf2005,Zhang2015}), little work appears to have been done on estimating the individual node degrees. Our work addresses this gap.  Formally, our interest lies in estimation of the degree of a vertex, provided that vertex is selected in a sample of the underlying graph. 

There are many sampling designs for graphs.  See~\cite[Ch 5]{Kolaczyk2009} for a review of the classical literature, and~\cite{ahmed2014network} for a recent survey.  Canonical examples include ego-centric sampling\cite{Handcock2010}, snowball sampling, induced/incident subgraph sampling, link-tracing and random walk based methods\cite{Leskovec2006,Ribeiro2010}. Under certain sampling designs where one observes the true degree of the sampled node (e.g. ego-centric and one-wave snowball sampling), degree estimation is unnecessary. In this paper, we focus on \emph{induced subgraph sampling}, which is structurally representative of a number of other sampling strategies\cite{Zhang2015}.  Formally, in induced subgraph sampling, a set of nodes is selected according to independent Bernoulli($p$) trials at each node. Then, the subgraph induced by the selected nodes, i.e., the graph generated by selecting edges between selected nodes, is observed.  This method of sampling shares stochastic properties with incident subgraph sampling (wherein the role of nodes and edges is reversed) and with certain types of random walk sampling~\cite{Ribeiro2010}.

The problem of estimating degrees of sampled nodes has been given a formal statistical treatment in~\cite{Zhang2005}, for the specific case of traceroute sampling as a special case of the so-called \emph{species problem} \cite{Bunge1993}. To the best of our knowledge, a similarly formal  treatment has not been applied more generally for other, more canonical sampling strategies.  However, a similar problem would be estimating personal network size for a group of people in a survey. Some prior works in this direction \cite{Killworth1998,McCormick2010} consider estimators obtained by scaling up the observed degree in the sampled network, in the spirit of what we term a method of moments estimator below. But no specific graph sampling designs are discussed in these studies. We focus on formulating the problem using the induced subgraph sampling design and exploit network information beyond sampled degree to propose estimators that are better than naive scale-up estimators.  Key to our formulation is a risk theoretic framework used to derive our estimators of the node degrees, through minimizing frequentist or Bayes risks. This contribution is accompanied by a comparative analysis of our proposed estimators and naive scale-up estimators, both theoretical and empirical, in several network regimes.

We note that when sampling is coupled with false positive and false negative edges, e.g., in certain biological networks, our methods are not immediately applicable. Sampling designs that result in the selection of a fraction of edges from the underlying global network (induced and incident subgraph sampling, random walks etc.) are our primary objects of study. We use induced subgraph sampling as a rudimentary but representative model for this class and aim to simultaneously estimate the true degrees of all the observed nodes with a precision better than that obtained by trivial scale-up estimators with no network information used.

\section{Degree Estimation Methods}
Let us denote by $G^0 = \left(V^0, E^0\right)$ a true underlying network, where $V^0 = \{1, \cdots, N\}$.  This network is assumed static and, without loss of generality, undirected.  The true degree vector is ${\bf d^0} = (d^0_1, \cdots, d^0_N)^T$. The sampled network is denoted by $G^* = \left(V^*, E^*\right)$ where, again without loss of generality, we assume that $V^* =\{1,\cdots,n\}$.  Write the sampled degree vector as ${\bf d^*} = \left(d^*_1, \cdots, d^*_n\right)$. Throughout the paper, we assume that we have an induced subgraph sample, with (known) sampling proportion $p$. 

It is easy to see from the sampling scheme that $d^*_i \sim B(d^0_i, p)$. Therefore, the method of moments estimator (MME) for $d^0_i$ is $\hat{d}^{\rm MME}_i = \frac{d^*_i}{p}$. Thus, $\hat{\bf d}_{\rm MME} = \left(\hat{d}^{\rm MME}_1,\cdots,\hat{d}^{\rm MME}_n\right)^T$ is a natural scale-up estimator of the degree sequence of the sampled nodes. In this section, we propose a class of estimators that minimize the unweighted $\ell_2$-risk of the sampled degree vector and discuss their theoretical properties. We aim to demonstrate, under several conditions, that the risk minimizers are superior to the regular scale-up estimators, the former taking into account the inherent relationships inside the network.

We note that although a maximum likelihood approach to estimation is perhaps intuitively appealing, a closed form derivation of the MLE in this setting is probitive. Another option is to look at marginal likelihoods. But the MLE based on univariate marginal likelihoods are essentially equivalent to the MME for this sampling scheme. We will frequently use the the first and second moments of the sampled degree vector in our estimation methods. The following lemma will be useful.

\begin{lem}\label{lem:meancovariancedegree}
Under induced subgraph sampling, the mean and covariance matrix of the observed degree vector are
\begin{align}
{\rm E}\left({\bf d}^*\right) &= p {\bf d^0} \\
{\rm Var}\left({\bf d}^*\right) &= p(1-p) {\mathscr D^0}
\end{align}
where the diagonals of ${\mathscr D^0}$ are $d^0_1, \cdots, d^0_n$ and the $(i,j)$-th off-diagonal is denoted by $d^0_{ij}$, which denotes the number of common neighbors of node $i$ and node $j$ in the network $G^0$.
\end{lem}

\subsection{Frequentist Risk Minimization}
Adopting the standard definition of (unweighted) frequentist $\ell_2$ risk of an estimator $\hat{\theta}$ of a parameter $\theta_0$, i.e., ${\cal R}(\hat{\theta},\theta_0) = \mathbb{E}||\hat{\theta} - \theta_0||^2$, the frequentist risks are calculated for a general class of estimators. We also define ${\cal R}_{\cal A} (\hat{\theta},\theta_0) :=  \mathbb{E}\left(||\hat{\theta} - \theta_0||^2 {\bf 1}(G^* \in {\cal A})\right)$, a \emph{restricted risk function} assuming the sampled graph $G^*$ is restricted to some class ${\cal A}$. Our proposed candidates are the elements in the class of linear functions of the observed degree vector that minimize the risk or the restricted risk w.r.t. some class. It is expected that the optimal estimator will be a function of the parameter and hence another (naive) estimator will need to be plugged in.  Our final estimate will then be a plug-in risk minimizer.

\subsubsection{Univariate Risk Minimization}\label{sec:uniriskmin}

Here we estimate the node degrees individually, assuming that the estimate for the $i^{\rm th}$ node is of the form $\hat{d}_i = c_i d^*_i$, where $c_i$ is a scalar and $d^*_i$ is the observed degree in the sample. Since $d^*_i \sim B(d^0_i, p)$, where $d^0_i$ is the true degree of the $i^{\rm th}$ node,
$$
{\cal R}(\hat{d}_i, d^0_i)  = {\rm Bias}^2(c_i d^*_i) + {\rm Var}(c_i d^*_i)
			     = (c_i p d^0_i - d^0_i)^2 + p(1-p)c_i^2 d^0_i \enskip .$$
Differentiating w.r.t. $c_i$ and equating to 0, we get the optimal $c^*_i = \frac{d^0_i}{p d^0_i + 1 - p} $. 
Plugging in the MME of $d^0$, we get the plug-in univariate risk minimizer
		$ \hat{d}_{i,{\rm u, P}} = \frac{d^{*^2}_i}{p(d^*_i + 1 - p)}.$ 

Taylor expanding the above formula (during Taylor expansions of functions of $d^*_i$, we will assume that $d^*_i$ is concentrated around its mean, so that the Taylor expanded approximation is close) and taking expectation, we see that

$\mathbb{E}\left(\hat{d}_{i,{\rm u, P}}\right) = \mathbb{E}\left[\frac{d^{*^2}_i}{p(d^*_i + 1 - p)}\right] 
= \frac{1}{p} \mathbb{E}\left[ d^*_i\left(1 + \frac{1-p}{d^*_i}\right)^{-1} \right]
\approx  \frac{1}{p} \mathbb{E}\left[ d^*_i\left(1 - \frac{1-p}{d^*_i}\right) \right] = d^0_i - \frac{1-p}{p} \enskip . $ \\

The above calculation suggests that an adjustment needs to be made to $\hat{d}_{i,{\rm u, P}}$ by bias-correction, so that its risk becomes comparable to that of $\hat{d}^{\rm MME}_i$. In fact, we will show in Proposition \ref{unirisk} that our bias-corrected plug-in estimator has a lower risk than MME when the true degree is bigger than a lower bound, which can be expressed as a closed form function of the sampling proportion. Ultimately, our proposed univariate risk minimizer is given by 
\begin{align}
\hat{d}_{i,{\rm u}} = \frac{d^{*^2}_i}{p(d^*_i + 1 - p)} + \frac{1-p}{p} \label{uniriskminest}
\end{align}

\subsubsection{Multivariate Risk Minimization}

We extend the idea presented in the previous section to the multivariate case, in order to minimize the overall $\ell_2$ sum over all sampled nodes. The rationale for this extension is to exploit the covariance structure we derived in Lemma \ref{lem:meancovariancedegree} in estimating the degree vector. Accordingly, we consider all estimates of the form ${\bf \hat{d}} = A{\bf d^*}$, where $A$ is an $n\times n$ matrix. Using Lemma \ref{lem:meancovariancedegree}, we get the $\ell_2$ risk \\
$R({\bf \hat{d}}, {\bf d^0}) = (pA - I) {\bf d^0} {\bf d^{0^T}} (pA - I)^T + p(1-p) A\mathscr{D}^0 A^T  \\ \hspace*{0.52in} = A\left( p^2 {\bf d^0} {\bf d^{0^T}} A^T + p(1-p)\mathscr{D}^0 \right)A^T - p\left({\bf d^0} {\bf d^{0^T}} A^T + A {\bf d^0} {\bf d^{0^T}}\right) + {\rm constant}$ . \\
The multivariate risk minimizer is defined as \\
\centerline{	$A^* = {\rm argmin}_A \sum^n_{i=1} \mathbb{E}\left(\hat{d}_i - d^0_i\right)^2 = {\rm argmin}_A {\rm tr}\left(R(\hat{d}, d^0)\right)$ .}\\
Differentiating the objective function w.r.t. $A$ and equating it to $0$, we get \\
\centerline{$A^* = p {\bf d^0} {\bf d^{0^T}} \left(p^2 {\bf d^0} {\bf d^{0^T}} + p(1-p)\mathscr{D}^0\right)^{-1} \enskip .$} \\
Plugging in the MME of ${\bf d^0}$ and $\mathscr{D}^0$, we get the plug-in multivariate risk minimizer
\begin{align}
{\bf \hat{d}}_{\rm m} = \frac{1}{p} {\bf d^*} {\bf d^{*^T}} \left( {\bf d^*} {\bf d^{*^T}} + \mathscr{D}^*\right)^{-1} {\bf d^*} \enskip ,
\end{align}
where $d^*_{ij}$ denotes the number of common neighbors of node $i$ and node $j$ in the sample, and $\mathscr{D}^*$ is given by a matrix whose diagonals are $d^*_i$ and whose off-diagonals are $d^*_{ij}$, $i,j \in \{1, \cdots, n\}, \, i\ne j$.

\subsection{Bayes Risk Minimization}
In this section, we propose a Bayesian solution to our estimation problem, by putting a prior on the degree distribution. The principal motivation behind this approach is the desire to incorporate additional information on global network structure, where the natural candidate in this context is the degree distribution.  In case such a subjective prior is not available, an estimate of the degree distribution may be used. We propose and analyze estimators based on both known (subjective) and estimated degree distributions below.

First, let us assume that we know the degree distribution $\pi(\cdot)$ of the underlying network. Under the assumption that the true degree of node $i$ follows $\pi(\cdot)$, and under induced subgraph sampling of $G$, the conditional distribution of $d^*_i | d_i$ is $B(d_i,p)$. Then it can be easily shown that the Bayes estimator under square error loss is 
\begin{align}\label{eq:Bayes}
\hat{d}^{B}_i &= \frac{\sum_{d_i \geq d^*_i} d_i \binom{d_i}{d^*_i} (1-p)^{d_i} \pi(d_i)}{\sum_{d_i \geq d^*_i} \binom{d_i}{d^*_i} (1-p)^{d_i} \pi(d_i)} \enskip .
\end{align}

If the true degree distribution is not known, then it needs to be estimated, for example using techniques described in or similar to \cite{Zhang2015}. Let $\hat{\pi}(\cdot)$ be a ``reasonable" estimator for $\pi(\cdot)$. Then an empirical Bayes estimator is given by
\begin{align}\label{eq:EBayes}
\hat{d}^{EB}_i &= \frac{\sum_{d_i \geq d^*_i} d_i \binom{d_i}{d^*_i} (1-p)^{d_i} \hat{\pi}(d_i)}{\sum_{d_i \geq d^*_i} \binom{d_i}{d^*_i} (1-p)^{d_i} \hat{\pi}(d_i)} \enskip .
\end{align}
Generally speaking, if $\xi(d^*_i; d_i)$ denotes the distribution of $d^*_i$ given $d_i$, then this empirical Bayes estimate can be expressed as 
$$
\hat{d}^{EB}_i = \frac{\sum_{d_i \geq d^*_i} d_i \xi(d^*_i; d_i) \hat{\pi}(d_i)}{\sum_{d_i \geq d^*_i} \xi(d^*_i; d_i) \hat{\pi}(d_i)} \enskip .
$$

These estimators take the form of a weighted mean, as expected for Bayes estimates under quadratic loss. The weights are functionals of both sampling design and the degree distribution.  For the latter estimator, only the estimated degree distribution comes into play, and thus the proposed empirical Bayes estimator incorporates the sampling and sampled network information.

\section{Risk Analysis}
In this section, we present results on the relative performance of our proposed estimators from a risk-theoretic perspective, and we discuss several conditions under which one outperforms the other. All these estimates will be benchmarked against the regular scale-up estimate $\hat{\bf d}_{\rm MME}$.  Proofs may be found in the supplementary materials.

\subsection{Risk of Frequentist Estimates}
In the first part of our risk analysis, we look at the $\ell_2$ frequentist risk of our proposed univariate and multivariate estimators. Our main results in this section will compare the risk incurred by our proposed estimators to the scale up estimator and discuss conditions under which our proposed estimators perform better.

\begin{prop}\label{unirisk}
Assuming $d^0_i > \frac{1-p}{p}$, we have ${\cal R}\left(\hat{d}_{i,u},d^0_i\right) < {\cal R}\left(\hat{d}^{\rm MME}_i,d^0_i\right)$.
\end{prop}
In other words, the univariate risk minimizer $\hat d_{i,u}$ will outperform the MME when the true degree $d^0_i$ is 
sufficiently large.

\begin{prop}\label{multirisk}
Let us denote the class of all sampled graphs of size $n$ (where $d^*_i \geq 1$ for all $i$, i.e., there is no isolated node) as $\mathscr{G}^*_n $. Also assume that there exists an $0< \alpha_0 \leq 1 $ such that 
\begin{align*}
\mathscr{G}^*_{1,n} &= \left\{ {\cal G} \in \mathscr{G}^*_n : \text{ Normalized eigenvectors } {\bf v}_1, {\bf v}_2, \cdots, {\bf v}_n \text{ of }  \left( {\bf d^*} {\bf d^{*^T}} + \mathscr{D}^*\right) \text{ satisfy } \right. \\ 
&\qquad \qquad \qquad \qquad\qquad \qquad \qquad \qquad\left. {\bf 1^T}{\bf v}_i \geq \sqrt{n}\alpha_0 \quad \forall i \right\}\\
\mathscr{G}^*_{2,n} &= \left\{{\cal G} \in \mathscr{G}^*_n : \frac{n^3 \alpha^2_0}{|E({\cal G})|\left(\frac{2|E({\cal G})|}{n-1} + n\right)} \geq 1 -  \frac{(1-p)\lambda_{\rm min}(\mathscr{D})}{||{\bf d}^0||^2}\right\}
\end{align*} 
are nonempty. Then we have ${\cal R}_{\mathscr{G}^*_{1\cap2, n}}\left({\bf \hat{d}}_{\rm m},{\bf d}^0\right) \leq {\cal R}_{\mathscr{G}^*_{1\cap2, n}}\left(\hat{\bf d}^{\rm MME},{\bf d}^0\right)$ over sampled graphs belonging to $\mathscr{G}^*_{1\cap2, n} = \mathscr{G}^*_{1,n} \bigcap \mathscr{G}^*_{2,n}$.
\end{prop}

Scrutiny of the conditions in Proposition \ref{multirisk}, along with definition of the set $\mathscr{G}^*_{1\cap 2, n}$, reveals a general characterization of the graphs where the proposed multivariate estimator performs better. It is to be noticed that ${\bf \hat{d}}_{\rm m}$ shrinks $\hat{\bf d}^{\rm MME}$ by some factor. The term on the right side of the inequality in the definition of $\mathscr{G}^*_{2,n}$ provides a lower bound on the shrinkage factor and the term on the left decreases as the cardinality of $E({\cal G})$ increases, i.e., the graph becomes less sparse. Hence, the proposed estimator can be expected to work better than the standard scale-up estimator under the assumption of sparsity of the sampled graph. This will also be demonstrated in the simulation section.

The eigenvector condition imposes a geometric constraint on the sample degree-degree matrix $\mathscr{D}^*$. What it essentially means is that the angle between the eigenvectors of $\left( {\bf d^*} {\bf d^{*^T}} + \mathscr{D}^*\right)$ and ${\bf 1}$ should be smaller than ${\rm arccos}(\alpha_0)$. Or, in other words, by selecting an $\alpha_0$ sufficiently small but positive, our class of sampled graphs are restricted where the associated matrix $\left( {\bf d^*} {\bf d^{*^T}} + \mathscr{D}^*\right)$ has eigenvectors at least ${\rm arcsin}(\alpha_0)$ angle away from any orthogonal direction to ${\bf 1}$. Thus, our estimator performs better for sparse graph satisfying a mild geometric condition.

\subsection{Risk of Bayes Estimate}
The performance of the Bayes estimators is evaluated here under several conditions and network paradigms. Note that these estimators are compared to the regular scale-up estimator with respect to their frequentist risk functions. We start with our estimator in its most general form and state conditions on the prior degree distribution that will ensure lower risk. From that, we assess its risk when the prior degree distribution is replaced with an appropriate estimate. We also explicitly derive the Bayes estimator for the Erd\"os-R\'enyi class of random graphs and state conditions under which the Bayes estimator yields lower risk than the scale-up estimator. 

\begin{prop}\label{suffrisk}
Let ${d}^0_i$ be the true degree of sample node $i$, and $d^*_i$, the observed degree. Denote by $\mathscr{G}^*_{\rm B}$ the class of sampled graphs where the following two conditions hold:
\begin{align}
&{\mathbb E}\left(\sum_{d_i \geq d^*_i} \pi^2(d_i)\right) \leq \frac{p(1-p)}{(N - 1 - d^0_i)^2} d^0_i \hspace{1cm} \text{ when } d^0_i \leq \frac{N-1}{2}\enskip  ; \, and \label{suffrisk:cond1}\\ 
&\left.\frac{\sum_{d_i \geq d^*_i} p\left(d^*_i, d_i\right) \pi(d_i)}{\sum_{d_i \geq d^*_i} p\left(d^*_i, d_i\right)}\right. \geq p \enskip , \label{suffrisk:cond2}
\end{align} 
where $p\left(d^*_i, d_i\right) =  \binom{d_i}{d^*_i} (1-p)^{d_i}$.  Then ${\cal R}_{\mathscr{G}^*_{\rm B}}\left(\hat{d}^B_i, d^0_i\right) \leq {\cal R}_{\mathscr{G}^*_{\rm B}}\left(\hat{d}^{\rm MME}_i, d^0_i\right)$ under induced subgraph sampling.
\end{prop}

The conditions (\ref{suffrisk:cond1}) and (\ref{suffrisk:cond2}) essentially constrain the tail behavior of the prior degree disbution. The first condition ensures that the tail decays at a rate such that it is not too ``thick'' and the second condition ensures that it is not too ``thin''. As $d^0_i$ becomes bigger, the RHS in condition (\ref{suffrisk:cond1}) becomes smaller and that is reminiscent of the sparsity property of the underlying graph, meaning that not a lot of nodes can have very high degree, an observation consistent with sparse graphs. On the other hand, the LHS in the condition (\ref{suffrisk:cond2}) can be interpreted as the mean of the tail probabilities weighted by the posterior distribution. This has to be bounded away from zero in order for the Bayes estimate to have lower risk than the MME.

In real problems, where the true degree distribution is unknown, one either has to choose $\pi$ subjectively or use the data to come up with a reasonable estimate. Estimating $\pi$ for a general case is beyond the scope of this paper and will not be discussed here. For our analysis, we will just assume that we have an estimate of the degree distribution at our disposal (e.g., \cite{Zhang2015}), denoted by $\hat{\pi}$. Using $\hat{\pi}$ will give us our proposed empirical Bayes estimate $\hat{d}^{\rm EB}_i$, the behavior of which can be described as follows.

\begin{prop}\label{pluginapprox}
Let $\hat{\pi}(\cdot)$ be an estimate of $\pi(\cdot)$ such that $\left\|\hat{\pi} - \pi\right\|_\infty < \epsilon.$ Then under assumption (\ref{suffrisk:cond2}), with $\pi$ replaced by $\hat{\pi}$, we have 
\begin{align}
&\left| \sum_{d_i \geq d^*_i}\binom{d_i}{d^*_i}(1-p)^{d_i}\hat{\pi}(d_i) - \sum_{d_i \geq d^*_i}\binom{d_i}{d^*_i}(1-p)^{d_i}\pi(d_i)\right| < \frac{ \epsilon (1-p)^{d^*_i}}{p^{d^*_i + 1}}  \label{pluginapprox:num} \\
&\left| \sum_{d_i \geq d^*_i}d_i\binom{d_i}{d^*_i}(1-p)^{d_i}\hat{\pi}(d_i) - \sum_{d_i \geq d^*_i}d_i\binom{d_i}{d^*_i}(1-p)^{d_i}\pi(d_i)\right| < \frac{ \epsilon (1-p)^{d^*_i}}{p^{d^*_i + 2}}(d^*_i + 1 - p)\label{pluginapprox:den}
\intertext{Thus, it follows that}
&\frac{\left|\hat{d}^{EB}_i - \hat{d}^B_i\right|}{\hat{d}^B_i} <  \frac{\epsilon (1-p)^{d^*_i}}{p^{d^*_i + 1}\sum_{d_i \geq d^*_i}d_i\binom{d_i}{d^*_i}(1-p)^{d_i}\pi(d_i)} + \frac{\epsilon (1-p)^{d^*_i}(d^*_i + 1 - p)}{p^{d^*_i + 2}\sum_{d_i \geq d^*_i}\binom{d_i}{d^*_i}(1-p)^{d_i}\pi(d_i)} \label{pluginapprox:ratio}
\end{align}
\end{prop}

It is easily seen that with the assumption \eqref{suffrisk:cond2}, the upper bound in \eqref{pluginapprox:ratio} can be simplified to 
\begin{align*}
\frac{\left|\hat{d}^{EB}_i - \hat{d}^B_i\right|}{\hat{d}^B_i} &<  \frac{\epsilon (1-p)^{d^*_i}}{d^*_i p^{d^*_i + 2}\sum_{d_i \geq d^*_i}\binom{d_i}{d^*_i}(1-p)^{d_i}} + \frac{\epsilon (1-p)^{d^*_i}(d^*_i + 1 - p)}{p^{d^*_i + 3}\sum_{d_i \geq d^*_i}\binom{d_i}{d^*_i}(1-p)^{d_i}}.
\end{align*}
Assuming a large network, the sum in the denominator can be approximated by $\frac{(1-p)^{d^*_i}}{p^{d^*_i+1}}$. Then the upper bound is 
$$\frac{\epsilon}{d^*_i p} + \frac{\epsilon(d^*_i+1-p)}{p^2} = \frac{\epsilon}{p}\left(\frac{1}{d^*_i} + \frac{d^*_i + 1 - p}{p}\right).$$
From the above discussion, it is evident that if $\epsilon = o(p^2/n)$, $\hat{d}^{EB}_i \approx \hat{d}^B_i$ for all $i$ and hence their risk functions will also be close. Thus, using Proposition \ref{suffrisk}, it is expected that ${\cal R}_{\mathscr{G}^*_{\rm B}}\left(\hat{d}^{\rm EB}_i, d^0_i\right) \lesssim {\cal R}_{\mathscr{G}^*_{\rm B}}\left(\hat{d}^{\rm MME}_i, d^0_i\right)$

\subsubsection{Illustration: Erd\"os-R\'enyi Graphs}
It is well known that the asymptotic degrees in Erd\"os-R\'enyi graph models follow a Poisson distribution, under standard conditions. In this section, we study the effects of using a Poisson prior degree distribution for large Erd\"os-R\'enyi graphs. The goal is to demonstrate the efficacy of the Bayesian approach compared to scale-up estimates as in the last section. However, studying specific models like Erd\"os-R\'enyi will give us more insight about the performance of the proposed Bayes estimate. In this scenario, the prior $\pi(\cdot)$ is given by 
$$\pi(d_i) = e^{-\lambda} \frac{\lambda^{d_i}}{d_i !} \enskip ,$$
where $\lambda$ is the prior mean. For a large Erd\"os-R\'enyi graph with number of nodes $N$ and edge probability $p_e$, $\lambda \approx Np_e$. We denote, by $P(k, \mu)$, the shifted Poisson distribution on $k, k+1, \cdots, \infty$ whose p.m.f. is given by
$$f(x) = e^{-\mu}\frac{\mu^{x-k}}{(x-k)!}\mathbf{1}_{\{k, k+1, \cdots\}}(x).$$ It is easy to check that with a ${\rm Poisson}(\lambda)$ prior on $d_i$, the posterior distribution is $P\left(d^*_i, \lambda(1-p)\right)$. Hence the Bayes estimate with respect to the quadratic loss function is $$\hat{d}^{\rm B}_i = d^*_i + \lambda(1-p)\enskip .$$

\begin{prop}\label{ErdosBayesRisk}
Assuming
$$\lambda + \frac{1+p}{p}\left( \frac{1}{2} - \sqrt{\frac{\lambda p}{1+p} + 1}\right) \leq d^0_i \leq \lambda + \frac{1+p}{p}\left( \frac{1}{2} + \sqrt{\frac{\lambda p}{1+p} + 1}\right),$$
the quadratic risk of the Bayes estimator using a ${\rm Poisson}(\lambda)$ prior is smaller than that of the MME.
\end{prop}

The above result shows that if the sampled node is such that its true degree belongs to a neighborhood around the mean of the underlying degree distribution, then the Bayes estimator is uniformly better than the MME. In case the underlying mean is unknown, it can easily be estimated from the sample. (e.g., for known $N$, $\hat{\lambda}_e = N\hat{p_e} = N |E(G^*)| / \binom{n}{2}.$) If $\hat{\lambda}$  is a consistent estimator of $\lambda$ in the sense that $\hat{\lambda} \stackrel{P}{\rightarrow} \lambda$ when $N \rightarrow \infty$, $n\rightarrow\infty$ and $n/N \rightarrow p$, then the empirical Bayes estimator $$\hat{d}^{\rm EB}_i = d^*_i + \hat{\lambda}(1-p)$$ will converge in probability to the Bayes estimator in the sense that $\left|\hat{d}^{\rm EB}_i - \hat{d}^{\rm B}_i\right| \stackrel{P}{\rightarrow} 0$. Hence, the result of Prop. (\ref{ErdosBayesRisk}) is expected to hold. This will also be demonstrated in the simulations.

\section{Simulations}
For our simulation study, we look at two different regimes of network -- Erd\"os-R\'enyi random graphs and heavy tailed degree distributions. 
\subsection{Erd\"os-R\'enyi network}
We compare four methods of estimation - the regular MME, univariate risk minimizer, multivariate risk minimizer and the Bayes estimate. As priors in Bayes estimation, we use both exponentially decaying (Poisson) and polynomially decaying degree distribution as priors. Table \ref{ERtable} records the Euclidean distance between the true and estimated degree vectors across some combinations of graph size $N$, edge strength $p_e$ and sampling proportion $p$. The errors are averaged over 50 different samples from each given graph $G$. From the output, it is clear that the Bayes estimators with true $\lambda$ and estimated $\lambda$ outperform other estimators by a very wide margin in terms of $\ell_2$ risk. Also, our theoretical prediction in the discussion following Proposition \ref{multirisk} was that the multivariate risk minimizer (MRM) works better than the MME for sparse graphs. This is experimentally verified in this simulation, since we see that the relative risk of MRM compared to MME decreases as the sparsity of the underlying graph increases, i.e., as $p_e$ decreases. The method with lowest total quadratic loss is shown in red for each condition.

\subsection{Scale Free Network}
We compared four methods of estimation in simulated scale free networks which follow a power law degree distribution. As priors in Bayes estimation, we compared the true polynomial prior and quadratic prior. We computed the $l_2$ distances across some combinations of sparsity (denoted by $s$, given by the ratio of total edges to all possible edges), sampling proportion $p$ and heaviness of the tail of the degree distribution, controled by $m$. The results are shown in Table \ref{SFtable}. The Bayes estimators or the multivariate risk minimizers work better than the other estimators. One important thing to observe here is that for the most sparse graph, the Bayes estimator with true prior works the best and as $s$ increases, multivariate risk minimizers work better than the rest, but there is hardly any improvement over MME. Again, the method with lowest total quadratic loss is shown in red for each condition.

\begin{table}[h]
\centering
\begin{minipage}[t]{0.49\hsize}\centering
\resizebox{\columnwidth}{!}{%
\begin{tabular}{|c|c|c|c|c|c|c|c|c|c|c|c|}
\toprule
$p_e, p\downarrow, N\rightarrow$ & \multicolumn{6}{|c|}{$N = 1000$}\\
\midrule
 & MME & URM & MRM & \multicolumn{3}{|c|}{Bayes}\\ \hline
 & & & & Pois.($\lambda$) & Pois.($\hat{\lambda}$)& Poly.\\ \hline
$p_e = 0.1, p = 0.1$ & 292.29 & 290.04 & 289.76 & \textcolor{red}{90.03} &  95.95 & 292.48\\ \hline
$p_e = 0.2, p = 0.1$ & 416.02 & 415.15 & 413.28 & \textcolor{red}{121.32} & 128.49 & 416.02\\ \hline
$p_e = 0.3, p = 0.1$ & 492.22 & 491.88 & 488.05 & \textcolor{red}{136.86} & 149.02 & 492.64\\ \hline
$p_e = 0.4, p = 0.1$ & 588.18 & 587.84 & 586.40 & \textcolor{red}{152.94} &  168.99 & 588.02\\ \hline\hline
$p_e = 0.1, p = 0.2$ & 284.08 & 283.67 & 282.76 & \textcolor{red}{119.87} &  122.73 & 284.24\\ \hline
$p_e = 0.2, p = 0.2$ & 389.15 & 389.07 & 386.87 & \textcolor{red}{164.30} &  166.84 & 389.55\\ \hline
$p_e = 0.3, p = 0.2$ & 485.09 & 485.07 & 481.82 & \textcolor{red}{187.43} &  190.55 & 485.63\\ \hline
$p_e = 0.4, p = 0.2$ & 527.37 & 527.28 & 527.68 & \textcolor{red}{205.47} &  210.42 & 527.07\\ \hline
\bottomrule
\end{tabular}
}
\caption{Erd\"os-R\'enyi Simulation Results: $\lambda$ is the true mean using known $p_e$. $\hat{\lambda}$ is the estimated mean using an estimate $\hat{p_e}$ of $p_e$.}\label{ERtable}
\label{table:sim1}
\end{minipage}
\hfill
\begin{minipage}[t]{0.49\hsize}\centering
\resizebox{\columnwidth}{!}{%
\begin{tabular}{|c|c|c|c|c|c|}
\toprule
$s, p\downarrow, N\rightarrow$ & \multicolumn{5}{|c|}{$N = 1000$}\\
\midrule
 & MME & URM & MRM & \multicolumn{2}{|c|}{Bayes}\\ \hline
 & & & & True Prior & Quad. Prior\\ \hline
$s = 0.2\%, p = 0.1, m = 2$ & 45.60 & 35.76 & 43.78 & \textcolor{red}{33.21} & \textcolor{red}{33.21}\\ \hline
$s = 1\%, p = 0.1, m = 2$ & 92.13 & 85.39 & 89.93 & \textcolor{red}{82.29} & \textcolor{red}{82.29}\\ \hline
$s = 5\%, p = 0.1, m = 2$ & 238.10 & 234.28 & 237.27 & \textcolor{red}{232.76} & \textcolor{red}{232.76}\\ \hline
$s = 0.2\%, p = 0.1, m = 2.5$ & 42.48 & 28.26 & 40.27 & \textcolor{red}{19.23} & 21.07\\ \hline
$s = 1\%, p = 0.1, m = 2.5$ & 92.91 & 82.89 & 91.50 & 81.93 & \textcolor{red}{78.72}\\ \hline
$s =5\%,  p = 0.1, m = 2.5$ & 210.04 & 214.70 & \textcolor{red}{208.22} & 231.68 & 219.55\\ \hline
$s = 0.2\%, p = 0.1, m = 3$ & 41.52 & 28.75 & 39.36 & \textcolor{red}{21.71} & 22.61\\ \hline
$s = 1\%, p = 0.1, m = 3$ & 89.40 & 79.98 & 88.07 & 83.39 & \textcolor{red}{75.46}\\ \hline
$s = 5\%, p = 0.1, m = 3$ & 209.97 & 213.30 & \textcolor{red}{208.25} & 242.90 & 217.87\\ \hline
\bottomrule
\end{tabular}
}
\caption{Scale Free Simulation Results}\label{SFtable}
\label{table:sim2}
\end{minipage}
\end{table}

\section{Human Trafficking Network}
In February 2015, the Defense Advanced Research Projects Agency (DARPA), an agency of the U.S. Department of Defense, announced the \emph{Memex} program in response to the use of the Internet in human trafficking, especially chat forums, advertisements and job services sections. DARPA-funded research determined the trafficking industry spent \$250M to post more than 60M advertisements over a two-year time frame\cite{DARPA}. Indexing and cross-referencing the ads with the same contact number, similar address or zip codes help identify and track the illegal trafficking activities. This leads to a massive background network structure where each node represents an advertisement and an edge between two nodes are created if they share certain features. It is not unreasonable to expect that, in surveillance of networks like this, sampling may well arise, either by choice or by circumstance.  We mimic this situation by pretending that this underlying network generated by the \emph{Memex} program is unknown to us and sampling it using induced subgraph sampling. The nodes associated with trafficking activities are flagged in the data. There are 31,248 nodes, of which 12,387 are flagged and there are 10,200,838 edges. Our goal was to estimate the true degrees of flagged nodes that we saw in our sample. We compared the $\ell_2$ distance of regular scale-up estimators, and our proposed univariate, multivariate and Bayes estimators. For the Bayes estimator, a number of polynomial priors were taken into consideration with varying degree of decay, denoted by $\alpha$. The results are shown in Table \ref{humantraffic}. Almost everything works better than the naive scale-up estimator in terms of total $\ell_2$ loss, although the relative improvement is more modest than in simulation.

\begin{table}[h]
\centering
\resizebox{0.75\columnwidth}{!}{%
\begin{tabular}{|c|c|c|c|c|c|c|c|}
\toprule
$p$ & MME & URM & MRM & \multicolumn{3}{|c|}{Bayes}\\ \hline
 & & & & $\alpha = -0.1$ & $\alpha = -0.5$ & $\alpha = -1$\\ \hline
$p = 0.005$ & 3451.364 & \textcolor{red}{3436.64} & 3447.24 & 3687.26 & 3541.94 & 3450.97\\ \hline
$p = 0.01$ & 3427.55 & \textcolor{red}{3397.71} & 3427.88 & 3451.86 & 3412.12 & 3428.59\\ \hline
$p = 0.02$ & 4462.937 & \textcolor{red}{4448.33} & 4461.64 & 4492.83 & 4450.71 & 4462.31\\ \hline
\bottomrule
\end{tabular}
}
\caption{Sampling from Human Trafficking Network}\label{HTtable}
\label{humantraffic}

\end{table}

\section{Discussion \& Future Research}
In this paper, we addressed the problem of estimation of true degrees of sampled nodes from an unknown graph. We proposed a class of estimators from a risk-theory perspective where the goal was to minimize the overall $\ell_2$ risk of the degree estimates for the sampled nodes. We considered estimators that minimize both frequentist and Bayes risk functions and compared the frequentist $\ell_2$ risks of our proposed estimator to the naive scale-up estimator. The basic objective of proposing these estimators was to exploit the additional network information inherent in the sampled graph, beyond the observed degrees. Our theoretical analyses, simulation studies and real data show clear evidence of superior performance of our estimators compared to MME, especially when the graph is sparse and the sampling ratio is low, mimicking the real-world examples.

There are a number of ways our current work could be extended. Firstly, a theoretical analysis of the Bayes estimators under priors for random graph models beyond Erd\"os-R\'enyi is desirable, although likely more involved.  Secondly, although induced subgraph sampling serves as a representative structural model for a certain class of adaptive sampling designs, the specific details of the sufficiency conditions discussed in this paper can be expected to vary slightly with the other sampling designs (e.g., incident subgraph or random walk designs) .  Finally, the success of the Bayesian method appears to rely heavily upon appropriate choice of prior distribution, as observed in our theoretical analysis and computational experiments. It would be of interest to explore the performance of the empirical Bayes estimate in conjunction with the nonparametric method of degree distribution estimation proposed in~\cite{Zhang2015}.  More generally, the method in~\cite{Zhang2015} can in principle be extended to estimate individual vertex degrees.  But the computational challenge of implementation and the corresponding risk analysis can be expected to be nontrivial.

\newpage


\section*{Supplementary material}

# Estimation of Vertex Degrees in a Sampled Network: Supplementary-A: Proofs

## 1 Proofs

### 1.1 Proof of Lemma 2.1

*Proof.* Let $S$ be the set of sampled nodes. See that $d_i^* = \sum_{k \in \text{Ne}_i} I(k \in S)$. Hence, $d_i^* \sim B\left(d_i^0, p\right)$.

$$
\begin{aligned}
E\left(d_i^* d_j^*\right) &= E\left[\left(\sum_{k \in \text{Ne}_i} I(k \in S)\right)\left(\sum_{l \in \text{Ne}_j} I(l \in S)\right)\right] \\
&= E\left[\left(\sum_{k \in \text{Ne}_i \cap \text{Ne}_j} I(k \in S)\right) + \left(\sum\sum_{(k,l) \in (\text{Ne}_i \cup \text{Ne}_j) \setminus (\text{Ne}_i \cap \text{Ne}_j)} I(k \in S) I(l \in S)\right)\right] \\
&= d_{ij}^0 p + \left(d_i^0 d_j^0 - d_{ij}^0\right) p^2
\end{aligned}
$$

Note that $d_{ij}^0$ is the cardinality of the first set of nodes (by its definition) and $(d_i^0 d_j^0 - d_{ij}^0)$ is that of the second. The probability that a node is selected in induced subgraph sampling is $p$ and since each node is selected independently, the joint probability that two nodes are selected is $p^2$. Hence,

$$\text{Cov}\left(d_i^*, d_j^*\right) = d_{ij}^0 p(1-p)$$

□

### 1.2 Proof of Proposition 3.1

Taking Taylor expansion up to 2nd order, we get

$$
\begin{aligned}
\mathbb{E}\left(\hat{d}_{i,\text{u}}\right) &= \frac{1-p}{p} + \frac{1}{p} \mathbb{E}\left[d_i^* \left(1 + \frac{1-p}{d_i^*}\right)^{-1}\right] \\
&\approx \frac{1-p}{p} + \frac{1}{p} \mathbb{E}\left[d_i^* - (1-p) + \frac{(1-p)^2}{d_i^*}\right] \approx d_i^0 + \frac{(1-p)^2}{p^2 d_i^0}
\end{aligned}
$$

We only consider Taylor expansion up to 2nd order because the expectation of higher order terms can be neglected assuming $d_i^0$ is sufficiently large. Hence, we get

$$\text{Bias}\left(\hat{d}_{i,\text{u}}, d_i^0\right) = \frac{(1-p)^2}{p^2 d_i^0}$$

Similarly, we approximate the variance by Taylor expansion and get

$$
\begin{aligned}
\text{Var}\left(\hat{d}_{i,\text{u}}\right) &\approx \frac{1}{p^2} \text{Var}\left(d_i^* + \frac{(1-p)^2}{d_i^*}\right) \\
&= \frac{1}{p^2}\left[\text{Var}\left(d_i^*\right) + (1-p)^4 \text{Var}\left(\frac{1}{d_i^*}\right) + 2(1-p)^2 \text{Cov}\left(d_i^*, \frac{1}{d_i^*}\right)\right]
\end{aligned}
$$

We use second order Taylor expansion to approximate the covariance

$$\text{Cov}\left(d_i^*, \frac{1}{d_i^*}\right) = 1 - \mathbb{E}(d_i^*)\mathbb{E}\left(\frac{1}{d_i^*}\right)$$

$$\approx 1 - d_i^0 p \left(\frac{1}{d_i^0 p} + \frac{p(1-p)d_i^0}{(d_i^0 p)^3}\right) = -\frac{1-p}{p d_i^0}$$

Thus,

$$\text{Var}\left(\hat{d}_{i,u}\right) \approx \frac{1}{p^2}\left[p(1-p)d_i^0 + \frac{(1-p)^5}{p^3 d_i^{0^3}} - \frac{2(1-p)^3}{p d_i^0}\right]$$

Therefore, with some algebra the risk minimizing condition can be simplified as following

$$\mathcal{R}\left(\hat{d}_{i,u}, d_i^0\right) - \mathcal{R}\left(\hat{d}_i^{\text{MME}}, d_i^0\right) \approx \frac{(1-p)^4}{p^4 d_i^{0^2}} + \frac{(1-p)^5}{p^5 d_i^{0^3}} - \frac{2(1-p)^3}{p^3 d_i^0} < 0$$

$$\Leftrightarrow 2p^2 d_i^{0^2} - p(1-p)d_i^0 - (1-p)^2 > 0$$

$$\Leftrightarrow d_i^0 > \frac{p(1-p) + \sqrt{p^2(1-p)^2 + 4.2p^2.(1-p)^2}}{2.2p^2} = \frac{1-p}{p}$$

### 1.3 Proof of Proposition 3.2

*Proof.* Denote by $\mu_1, \mu_2, \cdots, \mu_n$ the eigenvalues, and $\mathbf{v}_1, \mathbf{v}_2, \cdots, \mathbf{v}_n$ the corresponding normalized eigenvectors, of $\left(\mathbf{d}^* \mathbf{d}^{*T} + \mathscr{D}^*\right)$. Then, note that

$$\mathbf{d}^{*T}\left(\mathbf{d}^*\mathbf{d}^{*T} + \mathscr{D}^*\right)^{-1}\mathbf{d}^* = \sum_{i=1}^{n}\frac{1}{\mu_i}\left(\mathbf{d}^{*T}\mathbf{v}_i\right)^2 \geq \sum_{i=1}^{n}\frac{1}{\mu_i}\left(\mathbf{1}^{*T}\mathbf{v}_i\right)^2 \quad (\text{Since } d_i^* \geq 1)$$

$$\geq n\alpha_0^2 \sum_{i=1}^{n}\frac{1}{\mu_i} \quad (\text{Since } \|\mathbf{v}_i\|^2 = 1)$$

$$\geq n^3 \alpha_0^2 \left(\sum_{i=1}^{n}\mu_i\right)^{-1} \geq n^3 \alpha_0^2 \left[e^*\left(\frac{2e^*}{n-1} + n\right)\right]^{-1}.$$

The last inequality follows from

$$\sum_{i=1}^{n}\mu_i = \mathbf{tr}\left(\mathbf{d}^*\mathbf{d}^{*T} + \mathscr{D}^*\right) = \sum_{i=1}^{n}\left(d_i^* + d_i^{*2}\right) \leq 2e^* + e^*\left(\frac{2e^*}{n-1} + n - 2\right).$$

Here $e^*$ denotes the number of edges in the sampled graph. We use the result proved by Caen[?] for the upper bound on the degree sum of squares.

Therefore, our proposed multivariate estimator is a shrinkage estimator of the regular scale-up estimator and the shrinkage factor is bounded away from zero. Now, for a shrinkage estimator $c\hat{\mathbf{d}}_{\text{MME}}$, it can be shown using simple algebra that a sufficient condition for $c\hat{\mathbf{d}}_{\text{MME}}$ to have lower risk than $\hat{\mathbf{d}}_{\text{MME}}$ is

$$c \geq 1 - \frac{(1-p)\lambda_{\min}(\mathscr{D}^0)}{\|\mathbf{d}^0\|^2}$$

Thus, for all graphs in $\mathscr{G}_{1\cap 2,n}^*$, risk for our proposed estimator is less than that of the MME.

□

## 1.4 Proof of Proposition 3.3

*Proof.*

$$\mathcal{R}_\pi\left(\hat{d}_i^B, d_i^0\right) = \mathbb{E}\left(\frac{\sum_{d_i \geq d_i^*}(d_i - d_i^0)p(d_i^*, d_i)\pi(d_i)}{\sum_{d_i \geq d_i^*} p(d_i^*, d_i)\pi(d_i)}\right)^2$$

$$= \mathbb{E}\left(\frac{\sum_{d_i \geq d_i^*}(d_i - d_i^0)p(d_i^*, d_i)\pi(d_i)/\sum_{d_i \geq d_i^*} p(d_i^*, d_i)}{\sum_{d_i \geq d_i^*} p(d_i^*, d_i)\pi(d_i)/\sum_{d_i \geq d_i^*} p(d_i^*, d_i)}\right)^2$$

$$\leq \frac{1}{p^2}\mathbb{E}\left(\max_{d_i \geq d_i^*}(d_i - d_i^0)^2 \sum_{d_i \geq d_i^*} \pi^2(d_i)\right)$$

$$= \frac{1}{p^2}\mathbb{E}\left(\max\left\{(d_i^* - d_i^0)^2, (N - 1 - d_i^0)^2\right\} \sum_{d_i \geq d_i^*} \pi^2(d_i)\right)$$

If $d_i^* < 2d_i^0 - N + 1$, then $d_i^0 - d_i^* > N - 1 - d_i^0$. Otherwise, $d_i^0 - d_i^* \leq N - 1 - d_i^0$. Thus, the above

$$= \frac{1}{p^2}\mathbb{E}\left((d_i^* - d_i^0)^2 \sum_{d_i \geq d_i^*} \pi^2(d_i)\mathbf{1}_{(d_i^* < 2d_i^0 - N + 1)}\right.$$

$$\left. + (N - 1 - d_i^0)^2 \sum_{d_i \geq d_i^*} \pi^2(d_i)\mathbf{1}_{(d_i^* \geq 2d_i^0 - N + 1)}\right)$$

$$= \frac{1}{p^2}\left(\mathbb{E}_1 + \mathbb{E}_2\right)$$

where $\mathbb{E}_1$ and $\mathbb{E}_2$ denote the expectations of the individual summands.

If $d_i^0 \leq \frac{N-1}{2}$. Then it is easy to check that $\mathbb{E}_1 = 0$

$$\mathcal{R}_\pi\left(\hat{d}_i^B, d_i^0\right) = \frac{(N - 1 - d_i^0)^2}{p^2}\mathbb{E}\left(\sum_{d_i \geq d_i^*} \pi^2(d_i)\right)$$

$$\leq \frac{p(1-p)}{p^2}d_i^0 \qquad \text{by the condition in (7) of Proposition 3.3.}$$

□

## 1.5 Proof of Proposition 3.4

*Proof.* It is easy to see that

$$\left|\sum_{d_i \geq d_i^*}\binom{d_i}{d_i^*}(1-p)^{d_i}\hat{\pi}(\cdot) - \sum_{d_i \geq d_i^*}\binom{d_i}{d_i^*}(1-p)^{d_i}\pi(\cdot)\right| < \epsilon S$$

where

$$S = \sum_{d_i \geq d_i^*}\binom{d_i}{d_i^*}(1-p)^{d_i}.$$

Hence, we have

$$S = \sum_{d_i \geq d_i^*} \binom{d_i}{d_i^*} (1-p)^{d_i}$$

$$\leq S' = \sum_{d_i = d_i^*}^{\infty} \binom{d_i}{d_i^*} (1-p)^{d_i}$$

$$= \frac{(1-p)^{d_i^*}}{p^{d_i^*+1}} .$$

Similarly,

$$\left| \sum_{d_i \geq d_i^*} d_i \binom{d_i}{d_i^*} (1-p)^{d_i} \hat{\pi}(\cdot) - \sum_{d_i \geq d_i^*} d_i \binom{d_i}{d_i^*} (1-p)^{d_i} \pi(\cdot) \right| < \epsilon T$$

where

$$T = \sum_{d_i \geq d_i^*} d_i \binom{d_i}{d_i^*} (1-p)^{d_i} .$$

Hence, we have

$$T = \sum_{d_i \geq d_i^*} d_i \binom{d_i}{d_i^*} (1-p)^{d_i}$$

$$\leq T' = \sum_{d_i = d_i^*}^{\infty} d_i \binom{d_i}{d_i^*} (1-p)^{d_i}$$

$$= -(1-p) \frac{d}{dp} \left[ \sum_{d_i = d_i^*}^{\infty} \binom{d_i}{d_i^*} (1-p)^{d_i} \right]$$

$$= -(1-p) \frac{d}{dp} \left[ \frac{(1-p)^{d_i^*}}{p^{d_i^*+1}} \right]$$

$$= \frac{(1-p)^{d_i^*} (d_i^* + 1 - p)}{p^{d_i^* + 2}} .$$

The last result follows easily from the above two. □

## 1.6 Proof of Proposition 3.5

*Proof.*

$$\mathcal{R}\left(\hat{d}_i^{BP}, d_i^0\right) = \text{Bias}^2\left(\hat{d}_i^{BP}\right) + \text{Var}\left(\hat{d}_i^{BP}\right)$$
$$= \left(\lambda - d_i^0\right)^2 (1-p)^2 + d_i^0 p(1-p)$$

Therefore,

$$\mathcal{R}\left(\hat{d}_i^{BP}, d_i^0\right) - \mathcal{R}\left(\hat{d}_i^{MME}, d_i^0\right)$$

$$= \left(\lambda - d_i^0\right)^2 (1-p)^2 + d_i^0 p(1-p) - \frac{d_i^0(1-p)}{p}$$

$$= (1-p) \left[ \left(\lambda - d_i^0\right)^2 (1-p) + d_i^0 p - \frac{d_i^0}{p} \right]$$

$$= (1-p)^2 \left[ \left(\lambda - d_i^0\right)^2 - \frac{d_i^0(1+p)}{p} \right]$$

$$= (1-p)^2 \left[ d_i^{0^2} - \left(2\lambda + \frac{1+p}{p}\right) d_i^0 + \lambda^2 \right] .$$

Hence, $\mathcal{R}\left(\hat{d}_i^{BP}, d_i^0\right) \leq \mathcal{R}\left(\hat{d}_i^{MME}, d_i^0\right)$ iff $d_i^0$ lies in between the roots of the quadratic equation $x^2 - \left(2\lambda + \frac{1+p}{p}\right)x + \lambda^2 = 0$, i.e.,

$$\frac{1}{2}\left(2\lambda + \frac{1+p}{p} - \sqrt{\left(2\lambda + \frac{1+p}{p}\right)^2 - 4\lambda^2}\right) \leq d_i^0 \leq \frac{1}{2}\left(2\lambda + \frac{1+p}{p} + \sqrt{\left(2\lambda + \frac{1+p}{p}\right)^2 - 4\lambda^2}\right)$$

Simplifying,

$$\lambda - \frac{1+p}{p}\left(\sqrt{\frac{\lambda p}{1+p} + 1} - \frac{1}{2}\right) \leq d_i^0 \leq \lambda + \frac{1+p}{p}\left(\sqrt{\frac{\lambda p}{1+p} + 1} + \frac{1}{2}\right) \ .$$

$\square$



\newpage

\end{document}